\documentclass[journal,a4paper]{IEEEtran}
\usepackage{graphicx}
\usepackage{amsmath,amssymb}
\usepackage{amssymb}
\usepackage{amsthm,amsmath}
\usepackage{lineno}
\usepackage{verbatim}
\usepackage{stmaryrd}
\usepackage{amsfonts}
\usepackage{graphicx,times,amsmath}
\usepackage{graphics,color}
\usepackage{graphicx}
\usepackage{latexsym}
\usepackage{amsmath}
\usepackage{amsfonts}
\usepackage{amssymb}
\usepackage{url}
\usepackage{subfigure}
\usepackage{times}
\usepackage{color}
\usepackage{verbatim}
\newtheorem{theorem}{Theorem}

\newtheorem{definition}{Definition}
\newtheorem{lemma}{Lemma}
\newtheorem{remark}{Remark}

\newtheorem{assumption}{Assumption}

\allowdisplaybreaks

\begin{document}
\title{Adaptive Finite Time Stability of Delayed Systems via Aperiodically Intermittent Control and Quantized Control}

\author{Xiwei Liu,~\IEEEmembership{Senior Member, IEEE}, Hailian Ma
\thanks{This work was supported by the National Science Foundation of China under Grant No. 61673298; Shanghai Rising-Star Program of China under Grant No. 17QA1404500; Natural Science Foundation of Shanghai under Grant No. 17ZR1445700; the Fundamental Research Funds for the Central Universities of Tongji University.}
\thanks{Department of Computer Science and Technology, Tongji University, and with the Key Laboratory of Embedded System and Service Computing, Ministry of Education, Shanghai 201804, China. Corresponding author: Xiwei Liu. Email:xwliu@tongji.edu.cn}}

\maketitle
\begin{abstract}
In this brief, we set up the finite time stability (FnTSta) theory for dynamical systems with bounded time-varying delays via aperiodically intermittent control (AIC) and quantized control (QC). A more general QC is designed in this brief. The bound of time-varying delay is required to be less than the infimum in AIC. Two-phases-method (2PM) is applied to solve the FnTSta for delayed system, i.e., it divides the whole proof process into two phases, one phase is that the process for the norm of system error evolving from initial values to $1$, and the other is the process for the norm of system error evolving from $1$ to $0$. By proving that these two phases both use FnT to realize, the whole FnTSta for system is proved. We also design the adaptive rules and prove its validity rigorously. Furthermore, the obtained theories are used to discuss the finite time synchronization (FnTSyn) of neural networks as an application. Finally, some simulations are given to illustrate the effectiveness of our theoretical results.
\end{abstract}

\section{Introduction}
In real world engineering applications, systems are preferred to reach stability in finite time instead of infinite time (such as exponential stability (ExSta) and asymptotic stability), since it has faster convergence rate and better performance. Finite time stability (FnTSta) problems (\cite{BB2000}-\cite{mr2018}) in dynamical systems have been extensively investigated, whose settling time (after this time point, the system error will be zero) depends on initial state. If the settling time is independent of initial state, then it is called fixed time stability (FxTSta, \cite{P2012}-\cite{luliuchen2016}). Due to the weak coupling among agents, the dynamical systems in real world can not achieve stability automatically. Therefore, it is necessary to design suitable spatiotemporal controllers for dynamical systems. As an important application of FnTSta/FxTSta, network finite (fixed) time synchronization (FnTSyn) has attracted more and more attention. For example, \cite{liujiang2016} investigated FnTSyn of memristor-based Cohen-Grossberg neural networks (NNs); \cite{liuchen2018} designed nonlinear coupling functions with or without control to realize FnTSyn and FxTSyn; FnT/FxT bipartite consensus was discussed in \cite{liucao2019}; etc.

{\bf Temporal control}: aperiodically intermittent control (AIC). Intermittent control (IC), where the control is discontinuous, has attracted interests from various areas. It is more economic and applicable in real world compared with continuous control strategies. Considering the wind power and solar power systems as real examples, AIC (\cite{liuchen2015a}-\cite{liuchen2015c}) was proposed as an improvement of periodically intermittent control (PIC, \cite{huang2009}-\cite{fengyang2016}) by loosening the requirement of periodicity between work time and rest time. The type of synchronization under IC mainly concentrates on the exponential synchronization (ExSyn), see \cite{gan2017}-\cite{chenwang19}, and ExSyn under adaptive IC was also discussed, see \cite{hujiang2015}-\cite{caili2018}. FnTSyn via IC has become a hot topic recently, for example, \cite{mei2014} considered the FnTSyn via adaptive PIC for drive-response systems and designed an adaptive control. Moreover, time delay is inevitable in real world due to limited speed of information processing or propagation, many researchers investigated FnTSyn via IC with time delay. For example, \cite{jingchen2016}-\cite{ganxaosheng2019} used the following inequalities (Lemma 3 in \cite{mei2014})
\begin{align*}
\left\{
\begin{array}{ll}
\dot{V}(t)\le -\alpha V^{\eta}(t), &lT\le t<lT+\theta T,\\
\dot{V}(t)\le0,&t\in lT+\theta T\le t<(l+1)T
\end{array}
\right.
\end{align*}
to investigate the FnTSyn via PIC, and the external controller in work time was composed of delay and also exist in rest time (see Eq. (5) in \cite{jingchen2016}, Eq. (3) in \cite{xuwu2019} and Eq. (2.5) in \cite{ganxaosheng2019}). \cite{jiangli2019} used the following inequalities (see Lemma 4)
\begin{align*}
\left\{
\begin{array}{ll}
\dot{V}(t)\le -\alpha V^{\eta}(t)+\beta V(t), &lT\le t<lT+\theta T,\\
\dot{V}(t)\le\beta V(t),&lT+\theta T\le t<(l+1)T
\end{array}
\right.
\end{align*}
to investigate the semiglobal FnTSyn with stochastic disturbance via PIC, and the external controller in work time was composed of delay and disappeared in rest time (see Eq. (6)).

{\bf Spatio control}--Quantized control (QC). In real applications, the signals transmission is always limited due to the bandwidth and capacity/storage of communication channels, so signals should be quantized before transmission, including uniform quantization (UQ, \cite{Nedic2009}, \cite{zhangzhang2017}), stochastic quantization (SQ, \cite{wangliao2015}, \cite{tangzhou2017}), logarithmic quantization (LQ, \cite{fuxie2005}-\cite{zhangzeng2018}), etc. In \cite{zhangzhang2017}, UQ was used to solve leader-following consensus problem of linear dynamical systems. Consensus tracking problem with SQ in second-order multi-agent systems was proposed in \cite{tangzhou2017}. \cite{fu2008} discussed several types of quantization and revealed that LQ has better performance than linear quantization in cases of quantized feedback control, quantized state estimation and multiplicative noise. Therefore, more and more researchers have applied LQ to signal transmission. \cite{xuyang2018} combined LQ with pinning IC to realize FnTSyn of dynamical systems. \cite{zhangli2019} investigated FxTSyn of complex networks via LQ pinning control, and \cite{zhangyang2019} continued to explore controllers with and without LQ to realize FnTSyn and FxTSyn.

Therefore, it is interesting to consider the FnTSta/FnTSyn problem for systems with delay both under AIC and QC. {\bf Our motivation is to design the controller as simple as possible, i.e., it should not be defined by time delays}. Most previous papers consider the problem by designing the controller with delay, see \cite{jingchen2016}-\cite{jiangli2019}. As for controllers without delay, \cite{jingzhang2019} studied the FnTSyn of networks with constant delay and AIC; \cite{yangprop} considered FnTSyn for NNs with proportional delay and designed an adaptive rule. When the controller is designed without time delay, time delay will surely appear in the analysis of FnTSta problem. However, to our best knowledge, there are no systematic means or theories for this case, while the FnTSta have already been set up in \cite{BB2000}. In \cite{wangchen2018} and \cite{wangchen2019}, when the author studied the FnT anti-synchronization in master-slave neural networks with time-varying delays, they split the proof process into two stages, from initial values to $1$ and from $1$ to $0$. Inspired by them, in \cite{liuli2020}, we proposed the definition called two-phase-method (2PM) and applied it to the FnT anti-synchronization of complex-valued NNs with asynchronous and bounded time delays. Recently, we used 2PM to set up the FnTsta theory for delayed systems \cite{liuarxiv}, where the delay can be time-varying, unbounded and non-differentiable, we also applied it to study the FnTSyn of complex networks. {\bf Since FnTSta problem via AIC and QC with delays is more complex than that in \cite{liuarxiv}, how to solve it is our aim in this brief}.

Contributions of this brief can be summarised as follows:

{\bf (1) A more generale QC is designed.} This new QC includes standard LQ, and we also present the necessary condition for QC, and show how to construct the QC.

{\bf (2) A controller without delay is designed and its validity for FnTSta is proved rigorously by using 2PM}. The bounded time delay can be time-varying and non-differentiable, which should be less than the infimum of control time span, and we obtain the criteria for FnTSta using 2PM and $\infty$-norm.

{\bf (3) An adaptive controller without delay is designed and the validity of adaptive rules is also proved.} Adaptive control is powerful in unknown environment, but for FnTSta problem it is more difficult for systems with adaptive coefficients than that with constant coefficients. \cite{liuarxiv} set up the adaptive theory for FnTSta with continuous controller, so we also consider the adaptive FnTSta via AIC and QC by 2PM.

The rest is organized as follows. Section \ref{model} presents the definitions and assumptions of AIC and a general quantizer, a lemma to ensure ExSta under AIC is also given. In Section \ref{theory}, by using 2PM, we set up the FnTSta theory for delayed system with small delay which should be less than the infimum of control time span, we also propose an adaptive controller and prove the validity of adaptive rules rigorously. Section \ref{app} applies the derived theories to FnTSyn of master-salve coupled NNs. Simulations are conducted in Section \ref{simu} to illustrate the correctness of our theoretical results. Section \ref{con} concludes this brief finally.

\section{Preliminaries}\label{model}
At first, we introduce a generalized definition of QC.

\begin{definition}\label{zer}
$q(\cdot): {R}\to\Pi$ is a quantizer, where $\Pi=\{\pm \pi_{j}: j=0,\pm 1,\pm 2,...\}\cup\{0\}$ with $\pi_{0}>0$ and $\pi_{j}>\pi_{j+1}$. For all $v\in R$, $q(v)$ is defined as:
\begin{align}
q(v)=\left\{\begin{array}{lc}
\pi_{j};&\frac{1}{1+\gamma_{j}}\pi_{j}<v\leq\frac{1}{1-\gamma_{j}}\pi_{j}\\
0;&v=0\\
-q(-v);&v<0
\end{array}\right.
\end{align}
\end{definition}

\begin{remark}
Obviously, if $\gamma_j$ is a constant, then it becomes the standard LQ in \cite{xuyang2018}. Moreover, this new quantizer is different from that in \cite{gaohuijun}-\cite{zhangzeng2018}, where they used heterogeneous $\gamma_j$ to represent different quantization levels for each dimension, but in the same dimension, the quantizer was still a standard LQ.
\end{remark}

In the following, we will introduce the necessary condition, the generality and advantage of the above defined quantizer.

At first, when $q(v)=\pi_j$, then $\frac{1}{1+\gamma_{j}}\pi_{j}<v\leq\frac{1}{1-\gamma_{j}}\pi_{j}$. Since the quantizer function holds for any real number, to ensure that the real set $R$ is the disjoint union of all the intervals like $(\frac{1}{1+\gamma_{j}}\pi_{j}, \frac{1}{1-\gamma_{j}}\pi_{j}]$, one should require
\begin{align}\label{req}
\frac{1}{1-\gamma_{j+1}}\pi_{j+1}=\frac{1}{1+\gamma_{j}}\pi_{j}
\end{align}
Therefore,
\begin{assumption}\label{as1}
The following condition holds:
\begin{align}\label{nece}
\frac{\pi_{j+1}}{\pi_{j}}=\frac{1-\gamma_{j+1}}{1+\gamma_{j}}
\end{align}
\end{assumption}

In \cite{xuyang2018}, the logarithmic quantizer is used with $\pi_j=\rho^j\pi_0$ and $\gamma_j=\delta$, therefore, according to (\ref{nece}), we have
\begin{align}
{\rho}=\frac{1-\delta}{1+\delta}
\end{align}
i.e., $\delta=(1-\rho)/(1+\rho)$.

On the other hand, if one let the $\gamma_j$ be constant, then from (\ref{nece}), we have that $\pi_j/\pi_{j+1}$ will also be a constant, i.e., $\pi_j$ will a geometric sequence. Otherwise, $\gamma_j$ are not constant, we will prove that $\gamma_j$ will converge to a constant. Since $\Pi$ is a given set, when $\pi_j$ is a geometric sequence, suppose $\pi_{j+1}/\pi_j=\rho\in (0,1)$, then from (\ref{nece}), one can get
\begin{align*}
\gamma_{j+1}=-\rho\gamma_j+(1-\rho)
\end{align*}
Solving this recurrence equation, we have
\begin{align*}
\gamma_j=(-\rho)^{j-1}\gamma_1+\frac{(-\rho)^{j-1}-1}{-\rho-1}(1-\rho){\longrightarrow}\frac{1-\rho}{1+\rho}, ~~{j\rightarrow\infty}
\end{align*}

Therefore, to generate more flexible quantizer, one can use $\pi_j$, $\gamma_j$, $\gamma_{j+1}$ and Assumption \ref{as1} to deduce $\pi_{j+1}$ step by step. For example, using only two alternative constants $\gamma$ and $\gamma^{\prime}\in (0,1)$, we can generate a new sequence with $\pi_0$ as:
\begin{align}\label{newquan}
&\Pi=\{\cdots, \pm\rho_1^{-2}\rho_2^{-2}\pi_0, \pm\rho_1^{-1}\rho_2^{-2}\pi_0, \pm\rho_1^{-1}\rho_2^{-1}\pi_0,\pm\rho_2^{-1}\pi_0, \nonumber\\
&\pm\pi_0, \pm\rho_1\pi_0, \pm\rho_1\rho_2\pi_0, \pm\rho_1^2\rho_2\pi_0, \pm\rho_1^2\rho_2^2\pi_0, \cdots\}\cup\{0\}
\end{align}
where $\gamma_{2l}=\gamma$, $\gamma_{2l+1}=\gamma^{\prime}$, and
\begin{align}\label{rho1rho2}
\rho_1=\frac{1-\gamma^{\prime}}{1+\gamma}\in (0, 1),~~\mathrm{and}~~ \rho_2=\frac{1-\gamma}{1+\gamma^{\prime}}\in (0, 1)
\end{align}
More complex sequences $\Pi$ can also be generated similarly.

For any $v\in R$, from Definition \ref{zer}, there exists an index $j(v)$, such that $\frac{1}{1+\gamma_{j}}\pi_{j}<v\le\frac{1}{1-\gamma_{j}}\pi_{j}$, which is equivalent to
\begin{align*}
(1-\gamma_{j})v\le q(v)=\pi_{j}<(1+\gamma_{j})v
\end{align*}
so we can find a Filippov solution $\gamma (v)\in[-\gamma_{j},\gamma_{j})$ such that $q(v)=(1+\gamma(v))v$.

\begin{assumption}\label{inf}
For the above quantizer in Definition \ref{zer}, we suppose that there is a constant $\overline{\gamma}$, such that
\begin{align}
\sup_j\gamma_j\le \overline{\gamma}<1
\end{align}
\end{assumption}

\begin{definition}
AIC strategy is defined as
\begin{align}
[0,+\infty)=&[t_0, t_1]\cup(t_1, t_2)\cup[t_2, t_3]\cup(t_3, t_4)\nonumber\\
&\cup\cdots[t_{2k}, t_{2k+1}]\cup(t_{2k+1}, t_{2k+2})\cup\cdots\label{AIC}
\end{align}
where the $k$-th intermittent window $[t_{2k}, t_{2(k+1)})$ is composed of control time span (CTS) $[t_{2k}, t_{2k+1}]$ and rest time span (RTS) $(t_{2k+1}, t_{2k+2})$, and $t_0=0$.
\end{definition}

\begin{assumption}\label{absolute}
For AIC (\ref{AIC}),
\begin{align}\label{abaic}
\inf_{k=0,1,\cdots} (t_{2k+1}-t_{2k})=\underline{\theta},~\sup_{k=0,1,\cdots} (t_{2k+2}-t_{2k})=\overline{\theta},
\end{align}
where $\underline{\theta}$ and $\overline{\theta}$ mean the time length for the minimum value of CTS and the maximum value of each intermittent window.
\end{assumption}

\begin{assumption}\label{sad}
The upper bound of time-varying delays in this paper should be less than the infimum of CTS, i.e.,
\begin{align}\label{restrict}
\tau(t)\le\tau^{\bullet}<\underline{\theta}
\end{align}
\end{assumption}

\begin{lemma} (\cite{liuchen2015b})\label{large}
For continuous nonnegative function $p(t): R\to R$
\begin{align}\label{pt}
\left\{
\begin{array}{ll}
\dot{p}(t)\le -m_1p(t)+m_2p(t-\tau(t));&t_{2k}\le t\le t_{2k+1}\\
\dot{p}(t)\le m_3p(t)+m_2p(t-\tau(t));&t_{2k+1}<t<t_{2k+2}
\end{array}
\right.
\end{align}
where $m_1, m_2, m_3$ are nonnegative constants. Suppose Assumption \ref{absolute} and \ref{sad} hold, if $m_1>m_2$, and
\begin{align}
\sigma^{\bullet}(\underline{\theta}-\tau^{\bullet})-(m_2+m_3)(\overline{\theta}-
\underline{\theta})>0,
\end{align}
where $\sigma^{\bullet}>0$ is the unique positive solution of the equation
\begin{align}
\sigma^{\bullet}-m_1+m_2e^{\sigma^{\bullet}\tau^{\bullet}}=0,
\end{align}
then ExSta for $p(t)$ can be realized.
\end{lemma}

\section{FnTSta under AIC and QC using 2PM}\label{theory}
At first, we consider the following dynamical system
\begin{align}\label{standard}
\dot{e}_i(t)=\alpha_1 e_i(t)+\alpha_2 e_i(t-\tau_i(t)))+u_i(t),
\end{align}
where $e(t)=(e_{1}(t),\cdots,e_{n}(t))^{T}\in R^n$; time-varying delay $\tau_{i}(t)$ is bounded; $\alpha_1>0, \alpha_2>0$; controller $u_i(t)$ in AIC is defined by
\begin{align}\label{control}
u_i(t)=
\left\{
\begin{array}{l}
-\mathrm{sgn}(q(e_i(t)))(\alpha_3|q(e_i(t))|+\alpha_4);\\
~~~~~~~~~~~~~~~~~~~~~~~~~~~~~~t\in[t_{2k}, t_{2k+1}]\\
0; ~~~~~~~~~~~~~~~~~~~~~~~~~~~t\in(t_{2k+1},t_{2k+2})
\end{array}
\right.
\end{align}
where $\alpha_3>0, \alpha_4>0$; $q(\cdot)$ is the quantizer in Definition \ref{zer}.

\begin{definition} (\cite{BB2000})
The origin $e(t)=0$ is said to be a FnTSta equilibrium if the FnT convergence condition and Lyapunov stability condition hold globally.
\end{definition}

\subsection{FnTSta with bounded small delay using 2PM}
\begin{theorem}\label{smallthm}
For system (\ref{standard}) under AIC (\ref{control}) and QC (\ref{zer}), Assumptions \ref{as1}-\ref{sad} hold. FnTSta can be realized if there exist positive constants $\varpi_1$ and $\varpi_2$, such that
\begin{align}
&(1-\overline{\gamma})\alpha_3>\alpha_1+\alpha_2,\label{condition1}\\
&\sigma^{\bullet}(\underline{\theta}-\tau^{\bullet})-(\alpha_1+\alpha_2)(\overline{\theta}-\underline{\theta})>0,\label{condition2}\\
&\alpha_2-\alpha_4+\varpi_1<0,\label{condition4}\\
&\alpha_1+\alpha_2-\varpi_2<0,\label{condition5}\\
&(\varpi_1+\varpi_2)\tau^{\bullet}-\varpi_{1}\underline{\theta}+\varpi_{2}(\overline{\theta}-\underline{\theta})\le-\phi^{\bullet}<0,\label{condition6}
\end{align}
where $\sigma^{\bullet}>0$ is the unique positive solution of the equation
\begin{align}
\sigma^{\bullet}+\alpha_1-(1-\overline{\gamma})\alpha_3+\alpha_2e^{\sigma^{\bullet}{\tau}^{\bullet}}=0.\label{condition3}
\end{align}
\end{theorem}

\begin{proof}
We will use the 2PM in \cite{liuarxiv} to prove FnTSta of system (\ref{standard}) with AIC (\ref{control}) and QC (\ref{zer}). If
\begin{align}\label{initial}
\sup_{-\tau^{\bullet}\le s\le 0}\|e(s)\|_{\infty}=\sup_{-\tau^{\bullet}\le s\le 0}\max_{i=1,\cdots,n}|e_i(s)|>1,
\end{align}
we will prove from Phase I; otherwise, we will switch directly to Phase II.

{\bf Phase I}: We will prove that the $\|e(t)\|_{\infty}$ will decrease from initial values to $1$ in FnT.

For time $t$, suppose that the index $i(t)$ satisfying $|e_i(t)|=\|e(t)\|_{\infty}$. Differentiating $|e_i(t)|$,

{Case 1}: When $t\in[t_{2k}, t_{2k+1}]$:
\begin{align}
&\frac{d{|e_i(t)|}}{dt}=\mathrm{sgn}(e_{i}(t))\big[\alpha_1e_{i}(t)+\alpha_2e_{i}(t-\tau_{i}(t))\nonumber\\
&-\mathrm{sgn}(q(e_{i}(t)))(\alpha_3|q(e_{i}(t))|+\alpha_4)\big]\nonumber\\
\le&[\alpha_1-(1-\overline{\gamma})\alpha_3]|e_i(t)|+\alpha_2|e_i(t-\tau_i(t))|-\alpha_4\nonumber\\
\le&[\alpha_1-(1-\overline{\gamma})\alpha_3]|e_i(t)|+\alpha_2|e_i(t-\tau_i(t))|\label{controlspanresult}
\end{align}

{Case 2}: When $t\in (t_{2k+1}, t_{2k+2})$,
\begin{align}
&\frac{d{|e_i(t)|}}{dt}=\mathrm{sgn}(e_{i}(t))\big[\alpha_1e_{i}(t)+\alpha_2e_{i}(t-\tau_{i}(t))\big]\nonumber\\
\le&\alpha_1|e_i(t)|+\alpha_2|e_i(t-\tau_i(t))|\label{restspanresult}
\end{align}

According to Lemma \ref{large} and conditions (\ref{condition1}), (\ref{condition2}) and (\ref{condition3}), ExSta can be realized, there must exist a time $T_1$ (without loss of generality, we can let $T_1=t_{2k_1}$, i.e., from the $k_1$-th window in AIC), such that
\begin{align}\label{converge}
\sup_{t-\tau(t)\le s\le t}\|e(t)\|_{\infty}\le 1, t\ge T_1
\end{align}

{\bf Phase II}: We will prove that the $\|e(t)\|_{\infty}$ will evolve from $1$ to $0$ in FnT. At first, from condition (\ref{condition6}), we have
\begin{align}
&(\varpi_1+\varpi_2)\tau^{\bullet}-\varpi_{1}(t_{2k+1}-t_{2k})+\varpi_{2}(t_{2k+2}-t_{2k+1})\nonumber\\
\le&-\phi^{\bullet}<0,~~k\ge k_1\label{conditionnew}
\end{align}

For $t\in[t_{2k_1}, t_{2k_1+1}]$, we consider the function
\begin{align}
V_1(t)=\|e(t)\|_{\infty}+\varpi_{1}t
\end{align}
and the corresponding maximum-value function
\begin{align}
W_1(t)=\sup_{t-\tau^{\bullet}\le s\le t}V_1(s)
\end{align}

Obviously, $V_1(t)\le W_1(t)$. If $V_1(t)<W_1(t)$, then there must exist a constant $\varsigma_1>0$ such that $V_1(s)\le W_1(t)$ for $s\in (t, t+\varsigma_1)$, i.e., $W_1(s)=W_1(t)$ for $s\in (t,t+\varsigma_1)$.
	
Else if there exists a point $t\ge T_1$, $V_1(t)=W_1(t)$, for this time $t$, suppose that the index $i(t)$ satisfying $|e_i(t)|=\|e(t)\|_{\infty}$. Differentiating $V_1(t)$, we have
\begin{align}
&\dot{V}_1(t)=\frac{d}{dt}\big(|e_i(t)|+\varpi_1t\big)\nonumber\\
\le&(\alpha_1-(1-\overline{\gamma})\alpha_3)|e_i(t)|+\alpha_2|e_i(t-\tau_i(t))|-\alpha_4+\varpi_1\nonumber\\
\leq&\alpha_2-\alpha_4+\varpi_1<0
\end{align}

Therefore,
\begin{align*}
&\sup_{t-\tau^{\bullet}\le s\le t}\|e(s)\|_{\infty}+\varpi_{1}(t-\tau^{\bullet})\le W_1(t)\le W_1(T_1) \nonumber\\
\le&\sup_{t_{2k_1}-\tau^{\bullet} \le s \le t_{2k_1}}\|e(s)\|_{\infty}+\varpi_{1}T_1
\end{align*}
i.e., when $t=t_{2k_1+1}$,
\begin{align}
&\sup_{t_{2k_1+1}-\tau^{\bullet}\le s\le t_{2k_1+1}}\|e(s)\|_{\infty}\label{erbao}\\
\le &\sup_{t_{2k_1}-\tau^{\bullet} \le s \le t_{2k_1}}\|e(s)\|_{\infty}+\varpi_1\tau^{\bullet}-\varpi_{1}(t_{2k_1+1}-t_{2k_1})\nonumber
\end{align}

For $t\in(t_{2k_1+1}, t_{2(k_1+1)})$, we consider the function
\begin{align}
V_2(t)=\|e(t)\|_{\infty}-\varpi_{2}t
\end{align}
and the corresponding maximum-value function
\begin{align}
W_2(t)=\sup_{t-\tau^{\bullet} \le s \leq t}V_2(s)
\end{align}

Obviously, $V_2(t)\le W_2(t)$. If $V_2(t)<W_2(t)$, then there must exist a constant $\varsigma_2>0$ such that $V_2(s)\le W_2(t)$ for $s\in (t, t+\varsigma_2)$, i.e., $W_2(s)=W_2(t)$ for $s\in (t, t+\varsigma_2)$.
	
Else if there exists a point $t\ge t_{2k_{1}+1}$, $V_2(t)=W_2(t)$, for this time $t$, suppose that the index $i(t)$ satisfying $|e_i(t)|=\|e(t)\|_{\infty}$. Differentiating $V_2(t)$, we have
\begin{align}
&\dot{V}_2(t)=\frac{d}{dt}\big(|e_i(t)|-\varpi_2t\big)\\
\le&\alpha_1|e_i(t)|+\alpha_2|e_i(t-\tau_i(t))|-\varpi_2\le \alpha_1+\alpha_2-\varpi_2<0\nonumber
\end{align}
Therefore, combing with (\ref{erbao}), when $t=t_{2k_1+2}$,
\begin{align*}
&\sup_{t_{2k_1+2}-\tau^{\bullet} \le s \leq t_{2k_1+2}}\|e(s)\|_{\infty}-\varpi_{2}t_{2k_1+2}\nonumber\\
\le&W_2(t)\le W_2(t_{2k_1+1})\nonumber\\
=&\sup_{t_{2k_1+1}-\tau^{\bullet} \le s \leq t_{2k_1+1}}\big(\|e(s)\|_{\infty}-\varpi_{2}s\big)\nonumber\\
\le&\sup_{t_{2k_1+1}-\tau^{\bullet} \le s \leq t_{2k_1+1}}\|e(s)\|_{\infty}-\varpi_{2}(t_{2k_1+1}-\tau^{\bullet})\nonumber
\end{align*}
i.e.,
\begin{align*}
&\sup_{t_{2k_1+2}-\tau^{\bullet} \le s \leq t_{2k_1+2}}\|e(s)\|_{\infty}\nonumber\\
\le&\sup_{t_{2k_1}-\tau^{\bullet} \le s \le t_{2k_1}}\|e(s)\|_{\infty}+(\varpi_1+\varpi_2)\tau^{\bullet}\nonumber\\
&-\varpi_{1}(t_{2k_1+1}-t_{2k_1})+\varpi_{2}(t_{2k_1+2}-t_{2k_1+1})\nonumber\\
\le&\sup_{t_{2k_1}-\tau^{\bullet} \le s \le t_{2k_1}}\|e(s)\|_{\infty}-\phi^{\bullet}
\end{align*}
which means that after one intermittent window $[t_{2k_1}, t_{2(k_1+1)})$, the value $\sup_{t-\tau^{\bullet} \le s \le t}\|e(s)\|_{\infty}$ at the final point $t=t_{2(k_1+1)}$ decreases at least $\phi^{\bullet}$ than the value at the start point $t=t_{2k_1}$. With the same process, after $k^{\bullet}$ intermittent windows, the final value will satisfy
\begin{align}
&\sup_{t_{2(k_1+k^{\bullet})}-\tau^{\bullet} \le s \leq t_{2(k_1+k^{\bullet})}}\|e(s)\|_{\infty}\nonumber\\
\le&\sup_{t_{2k_1}-\tau^{\bullet} \le s \le t_{2k_1}}\|e(s)\|_{\infty}-\phi^{\bullet}k^{\bullet}\le 1-\phi^{\bullet}k^{\bullet}
\end{align}

Therefore, after $\lceil 1/\phi^{\bullet}\rceil$ intermittent windows, i.e., $e(t)=0$, for all $t\ge T_2$, where
\begin{align}
T_2\ge T_1+\lceil 1/\phi^{\bullet}\rceil\overline{\theta}
\end{align}
i.e., FnTSta for system (\ref{standard}) is proved completely.
\end{proof}

\begin{remark}
In fact, besides the $\infty$-norm, $1$-norm, $2$-norm and other norms can also be used to analyze FnTSta. Concrete norm is determined by the real questiones, here we omit it.
\end{remark}

\begin{remark}\label{add}
In fact, there can be multiple terms containing time delays, like $\sum_{j=1}^{j^{\star}}\alpha_2^{j}e_i(t-\tau_i^j(t))$, the proof process is similar to that in Theorem \ref{smallthm}, and in this case, $\alpha_2$ in Theorem \ref{smallthm} can be chosen as $\sum_{j=1}^{j^{\star}}\alpha_2^{j}$.
\end{remark}

\begin{remark}
For system (\ref{standard}) under AIC (\ref{control}) and QC (\ref{zer}), parameters $\alpha_1$, $\alpha_2$, $\overline{\theta}$, $\underline{\theta}$, $\tau^{\bullet}$ and $\overline{\gamma}$ are all fixed, from (\ref{condition2}), we define
\begin{align}\label{biyong1}
\sigma^{\bullet}>\underline{\sigma}=\frac{(\alpha_1+\alpha_2)(\overline{\theta}-\underline{\theta})}{(\underline{\theta}-\tau^{\bullet})}
\end{align}
Moreover, for this $\underline{\sigma}$, according to (\ref{condition3}), we can also define
\begin{align}\label{xiajie1}
\underline{\alpha}_3=\frac{\underline{\sigma}+\alpha_1+\alpha_2e^{\underline{\sigma}{\tau}^{\bullet}}}{(1-\overline{\gamma})}.
\end{align}

Obviously, larger $\alpha_3 (>\underline{\alpha}_3)$ will generate larger $\sigma^{\bullet} (>\underline{\sigma})$, which makes condition (\ref{condition2}) easier to be realized.

Furthermore, from (\ref{condition6}), we can also define
\begin{align}\label{xiajienew}
\varpi_1>\underline{\varpi}_1=\frac{(\alpha_1+\alpha_2)(\tau^{\bullet}+\overline{\theta}-\underline{\theta})}{\underline{\theta}-\tau^{\bullet}}
\end{align}
and from (\ref{condition4}), define
\begin{align}\label{xiajie2}
\alpha_4>\underline{\alpha}_4=\alpha_2+\underline{\varpi}_1
\end{align}

Similarly, larger $\alpha_4$ generates larger $\varpi_1$, which makes condition (\ref{condition6}) easier to be realized with larger $\phi^{\bullet}$.

Therefore, from these analysis, larger parameters $\alpha_3$ and $\alpha_4$ will be better for FnTSta with earlier settling time. In the next, we will apply the adaptive technique to realize this idea.
\end{remark}

\subsection{Adaptive FnTSta with bounded small delay using 2PM}
Next, we design the adaptive controller as
\begin{align}\label{controladaptive}
u_i(t)=
\left\{
\begin{array}{l}
-\mathrm{sgn}(q(e_i(t)))(\alpha_3(t)|q(e_i(t))|+\alpha_4(t));\\
~~~~~~~~~~~~~~~~~~~~~~~~~~~~~~t\in[t_{2k}, t_{2k+1}]\\
0; ~~~~~~~~~~~~~~~~~~~~~~~~~~~t\in(t_{2k+1},t_{2k+2})
\end{array}
\right.
\end{align}
where the adaptive rules are defined as:
\begin{align}\label{alpha3adatpive}
\alpha_3(t)=\left\{
\begin{array}{ll}
\alpha_3(0),&t=0\\
\alpha_3(t_{2k+1}),&t=t_{2k+2}\\
0,&t_{2k+1}<t<t_{2k+2}
\end{array}
\right.
\end{align}
\begin{align}
\alpha_4(t)=\left\{
\begin{array}{ll}
\alpha_4(0),&t=0\\
\alpha_4(t_{2k+1}),&t=t_{2k+2}\\
0,&t_{2k+1}<t<t_{2k+2}
\end{array}
\right.
\end{align}
with for $t_{2k}\le t\le t_{2k+1}$,
\begin{align}\label{alpha3adatpiverule}
\dot{\alpha}_3(t)=
\left\{
\begin{array}{ll}
\mu_1e^{\eta t}\|q(e(t))\|_{\infty};&\sup_{t-\tau^{\bullet}\le s\le t}\|e(t)\|_{\infty}> 1\\
\mu_2\|q(e(t))\|_{\infty};&\sup_{t-\tau^{\bullet}\le s\le t}\|e(t)\|_{\infty}\le 1
\end{array}
\right.
\end{align}
\begin{align}\label{alpha4adatpiverule}
\dot{\alpha}_4(t)=\left\{
\begin{array}{ll}
0;&\sup_{t-\tau^{\bullet}\le s\le t}\|e(t)\|_{\infty}> 1\\
\mu_3;&0<\sup_{t-\tau^{\bullet}\le s\le t}\|e(t)\|_{\infty}\le 1\\
0;&\sup_{t-\tau^{\bullet}\le s\le t}\|e(t)\|_{\infty}=0
\end{array}
\right.
\end{align}
where $\alpha_3(0)\ge 0$, $\alpha_4(0)\ge 0$, $\mu_1, \mu_2, \mu_3$ are positive scalars, parameter $\eta>0$ will be defined later.

\begin{theorem}
For the system (\ref{standard}) under QC and AIC with adaptive strategy (\ref{controladaptive})-(\ref{alpha4adatpiverule}), if $\eta>(1-\overline{\gamma})\underline{\alpha}_3-\alpha_1$, where $\underline{\alpha}_3$ is defined in (\ref{xiajie1}), then FnTSta can be realized.
\end{theorem}

\begin{proof}
2PM is also used to prove the FnTSta under adaptive rules.

{\bf Phase I}: Prove that the $\|e(t)\|_{\infty}$ will evolve from initial values to $1$ in FnT.

At first, according to (\ref{alpha3adatpive}) and (\ref{alpha3adatpiverule}), $\alpha_3(t)$ is a nondecreasing function, if from some time point $\overline{t}$, $\alpha_3(t)$ is larger than the value $\alpha_3^{\star}$ ($>\underline{\alpha}_3$ in (\ref{xiajie1})), which can make (\ref{condition1}) and (\ref{condition2}) hold, then according to the analysis in Phase I of Theorem \ref{smallthm}, ExSta will be realized.

Otherwise, $\alpha_3(t)<\alpha_3^{\star}$ holds for all $t$. For any time $t$, suppose that the index $i(t)$ satisfying $|e_i(t)|=\|e(t)\|_{\infty}$.

{Case 1}: When $t\in[t_{2k}, t_{2k+1}]$, choose the function
\begin{align}
\overline{V}(t)=\|e(t)\|_{\infty}+\frac{1}{2\mu_1}e^{-\eta t}(\alpha_3(t)-\alpha_3^{\star})^2
\end{align}
where $\eta=(1-\overline{\gamma})\alpha_3^{\star}-\alpha_1$.
Differentiating it, we have
\begin{align}
&\dot{\overline{V}}(t)=\mathrm{sgn}(e_{i}(t))\big[\alpha_1e_{i}(t)+\alpha_2e_{i}(t-\tau_{i}(t))\nonumber\\
&-\mathrm{sgn}(q(e_{i}(t)))(\alpha_3(t)|q(e_{i}(t))|+\alpha_4(t))\big]\nonumber\\
&-\eta\frac{1}{2\mu_1}e^{-\eta t}(\alpha_3(t)-\alpha_3^{\star})^2+(\alpha_3(t)-\alpha_3^{\star})|q(e_i(t))|\nonumber\\
\le&[\alpha_1-(1-\overline{\gamma})\alpha_3^{\star}]|e_i(t)|+\alpha_2|e_i(t-\tau_i(t))|\nonumber\\
&-\eta\frac{1}{2\mu_1}e^{-\eta t}(\alpha_3(t)-\alpha_3^{\star})^2\nonumber\\
\le&[\alpha_1-(1-\overline{\gamma})\alpha_3^{\star}]\overline{V}(t)+\alpha_2\overline{V}(t-\tau_i(t))\label{controlspanresultada}
\end{align}

{Case 2}: When $t\in (t_{2k+1}, t_{2k+2})$, choose the function
\begin{align}
\overline{V}(t)=\|e(t)\|_{\infty}+\frac{1}{2}e^{-\eta t}(\alpha_3(t_{2k+1})-\alpha_3^{\star})^2
\end{align}
Differentiating it, we have
\begin{align}
\dot{\overline{V}}(t)=&\mathrm{sgn}(e_{i}(t))\big[\alpha_1e_{i}(t)+\alpha_2e_{i}(t-\tau_{i}(t))\big]\nonumber\\
&-\eta\frac{1}{2}e^{-\eta t}(\alpha_3(t_{2k+1})-\alpha_3^{\star})^2\nonumber\\
\le&\alpha_1\overline{V}(t)+\alpha_2\overline{V}(t-\tau_i(t))\label{restspanresultada}
\end{align}

According to Lemma \ref{large}, ExSta for $\overline{V}(t)$ can be realized. According to the definition of $\overline{V}(t)$, one can also get that ExSta for $\|e(t)\|_{\infty}$ can be realized, i.e., $\|e(t)\|_{\infty}$ will decrease from initial values to $1$ in FnT. Without loss of generality, we also assume (\ref{converge}) holds.

{\bf Phase II}: Prove that the $\|e(t)\|_{\infty}$ will evolve from $1$ to $0$ in FnT.

The arguments for $\alpha_3^{\star}$ can also be applied on any sufficiently large constant $\alpha_4^{\star}$ ($>\underline{\alpha}_4$ in (\ref{xiajie2})), here we omit it. For fixed $\alpha_1$ and $\alpha_2$, we can choose a constant $\varpi_2>0$ such that (\ref{condition5}) holds. Thus, for this chosen $\alpha_4^{\star}$, we can pick a scalar $\varpi_1>0$ such that (\ref{xiajienew}) and (\ref{xiajie2}) hold, i.e.,
\begin{align}
&(\varpi_1+\varpi_2)\tau^{\bullet}-\varpi_{1}(t_{2k+1}-t_{2k})+\varpi_{2}(t_{2k+2}-t_{2k+1})\nonumber\\
\le&-\phi^{\bullet}<0,~~k\ge k_1\label{conditionnewada}
\end{align}

For $t\in[t_{2k_1}, t_{2k_1+1}]$, we consider the function
\begin{align}
V_3(t)=&\|e(t)\|_{\infty}+\frac{1}{2\mu_2}(\alpha_3(t)-\alpha_3^{\star})^2+\frac{1}{2\mu_3}(\alpha_4(t)-\alpha_4^{\star})^2\nonumber\\
&+\varpi_{1}t\label{fxt1}
\end{align}
where the parameters $\mu_2$ and $\mu_3$ are chosen to satisfy
\begin{align}
\frac{(\alpha_3^{\star})^2}{2\mu_2}+\frac{(\alpha_4^{\star})^2}{2\mu_3}-\phi^{\bullet}=-\phi^{\star}<0
\end{align}
and the corresponding maximum-value function
\begin{align}
W_3(t)=\sup_{t-\tau^{\bullet}\le s\le t}V_3(s)
\end{align}

Obviously, $V_3(t)\le W_3(t)$. If $V_3(t)<W_3(t)$, then there must exist a constant $\varsigma_3>0$ such that $V_3(s)\le W_3(t)$ for $s\in (t, t+\varsigma_3)$, i.e., $W_3(s)=W_3(t)$ for $s\in (t,t+\varsigma_3)$.
	
Else if there exists a point $t\ge T_1$, $V_3(t)=W_3(t)$, for this time $t$, suppose that the index $i(t)$ satisfying $|e_i(t)|=\|e(t)\|_{\infty}$. Differentiating $V_3(t)$, we have
\begin{align}
&\dot{V}_3(t)\nonumber\\
\le&(\alpha_1-(1-\overline{\gamma})\alpha_3^{\star})|e_i(t)|+\alpha_2|e_i(t-\tau_i(t))|-\alpha_4^{\star}+\varpi_1\nonumber\\
\leq&\alpha_2-\alpha_4^{\star}+\varpi_1<0
\end{align}

Therefore,
\begin{align*}
&\sup_{t-\tau^{\bullet}\le s\le t}\|e(s)\|_{\infty}+\varpi_{1}(t-\tau^{\bullet})\le W_3(t)=W_3(T_1) \nonumber\\
\le&\sup_{t_{2k_1}-\tau^{\bullet} \le s \le t_{2k_1}}\|e(s)\|_{\infty}+\varpi_{1}T_1+\frac{(\alpha_3^{\star})^2}{2\mu_2}+\frac{(\alpha_4^{\star})^2}{2\mu_3}
\end{align*}
i.e., when $t=t_{2k_1+1}$,
\begin{align}
&\sup_{t_{2k_1+1}-\tau^{\bullet}\le s\le t_{2k_1+1}}\|e(s)\|_{\infty}\nonumber\\
\le &\sup_{t_{2k_1}-\tau^\bullet \le s \le t_{2k_1}}\|e(s)\|_{\infty}\nonumber\\
&+\varpi_1\tau^{\bullet}-\varpi_{1}(t_{2k_1+1}-t_{2k_1})+\frac{(\alpha_3^{\star})^2}{2\mu_2}+\frac{(\alpha_4^{\star})^2}{2\mu_3}\label{erbaoada}
\end{align}

For $t\in(t_{2k_1+1}, t_{2(k_1+1)})$, we consider the functions $V_2(t)$ and $W_2(t)$, then with the same process in Theorem \ref{smallthm}, and combing with (\ref{erbaoada}), when $t=t_{2k_1+2}$,
\begin{align*}
&\sup_{t_{2k_1+2}-\tau^{\bullet} \le s \leq t_{2k_1+2}}\|e(s)\|_{\infty}-\varpi_{2}t_{2k_1+2}\nonumber\\
\le&W_2(t)=W_2(t_{2k_1+1})\nonumber\\
=&\sup_{t_{2k_1+1}-\tau^{\bullet} \le s \leq t_{2k_1+1}}\big(\|e(s)\|_{\infty}-\varpi_{2}s\big)\nonumber\\
\le&\sup_{t_{2k_1+1}-\tau^{\bullet} \le s \leq t_{2k_1+1}}\|e(s)\|_{\infty}-\varpi_{2}(t_{2k_1+1}-\tau^{\bullet})\nonumber
\end{align*}
i.e.,
\begin{align*}
&\sup_{t_{2k_1+2}-\tau^{\bullet} \le s \leq t_{2k_1+2}}\|e(s)\|_{\infty}\nonumber\\
\le&\sup_{t_{2k_1}-\tau^{\bullet} \le s \le t_{2k_1}}\|e(s)\|_{\infty}\nonumber\\
&+(\varpi_1+\varpi_2)\tau^{\bullet}+\frac{(\alpha_3^{\star})^2}{2\mu_2}+\frac{(\alpha_4^{\star})^2}{2\mu_3}\nonumber\\
&-\varpi_{1}(t_{2k_1+1}-t_{2k_1})+\varpi_{2}(t_{2k_1+2}-t_{2k_1+1})\nonumber\\
\le&\sup_{t_{2k_1}-\tau^{\bullet} \le s \le t_{2k_1}}\|e(s)\|_{\infty}+\frac{(\alpha_3^{\star})^2}{2\mu_2}+\frac{(\alpha_4^{\star})^2}{2\mu_3}-\phi^{\bullet}\nonumber\\
=&\sup_{t_{2k_1}-\tau^{\bullet} \le s \le t_{2k_1}}\|e(s)\|_{\infty}-\phi^{\star}
\end{align*}
which means that after one intermittent window $[t_{2k_1}, t_{2(k_1+1)})$, the value $\sup_{t-\tau^{\bullet} \le s \le t}\|e(s)\|_{\infty}$ at the final point $t=t_{2(k_1+1)}$ decreases at least $\phi^{\star}$ than the value at the start point $t=t_{2k_1}$. With the same process, after $k^{\star}$ intermittent windows, the final value will satisfy
\begin{align}
&\sup_{t_{2(k_1+k^{\star})}-\tau^{\bullet} \le s \leq t_{2(k_1+k^{\star})}}\|e(s)\|_{\infty}\nonumber\\
\le&\sup_{t_{2k_1}-\tau^{\bullet} \le s \le t_{2k_1}}\|e(s)\|_{\infty}-\phi^{\star}k^{\star}\le 1-\phi^{\star}k^{\star}
\end{align}

Therefore, after $\lceil 1/\phi^{{\star}}\rceil$ intermittent windows, i.e., $e(t)=0$, for all $t\ge T_2$, where
\begin{align}
T_2\ge T_1+\lceil 1/\phi^{{\star}}\rceil\overline{\theta}
\end{align}
i.e., FnTSta for system (\ref{standard}) is proved completely.
\end{proof}

\begin{remark}
Adaptive rules for $1$-norm and $2$-norm can also be designed, which can be found in \cite{liuarxiv}, here we omit it.
\end{remark}

\section{Finite time synchronization of neural networks with delays}\label{app}
As an application of FnTSta results in the last section, we will investigate FnTSyn for master-slave coupled DNNs.

The master system is
\begin{align}\label{master}
\dot{x}_{i}(t)=&-d_{i}x_{i}(t)+\sum_{j=1}^na_{ij}f_{j}(x_{j}(t))\nonumber\\
&+\sum_{j=1}^nb_{ij}g_{j}(x_{j}(t-\tau_{ij}(t)))+I_{i},
\end{align}
where $x_{i}(t)\in R$; $d_{i}>0$; $f_{j}(\cdot)$ and $g_j(\cdot)$ are $R\rightarrow R$ functions; $I_i$ are inputs.

The slave system is
\begin{align}\label{slave}
\dot{y}_{i}(t)=&-d_{i}y_{i}(t)+\sum_{j=1}^na_{ij}f_{j}(y_{j}(t))\nonumber\\
&+\sum_{j=1}^nb_{ij}f_{j}(g_{j}(t-\tau_{ij}(t)))+I_{i}+u_{i}(t),
\end{align}

The synchronization error is defined by $e_{i}(t)=y_{i}(t)-x_{i}(t), i=1,2,\cdots,n$, whose dynamics are
\begin{align}\label{NNerror}
&\dot{e}_i(t)\\
=&-d_{i}e_{i}(t)+\sum_{j=1}^na_{ij}\tilde{f}_j(e_j(t))+\sum_{j=1}^nb_{ij}\tilde{g}_j(e_j(t-\tau_{ij}(t)))\nonumber
\end{align}
where $\tilde{f}_j(e_j(t))=f_j(y_j(t))-f_j(x_j(t))$ and $\tilde{g}_j(e_j(t))=g_j(y_j(t))-g_j(x_j(t))$. We assume that for $j=1,\cdots,n$,
\begin{align}\label{lip}
|\tilde{f}_j(e_j(t))|\le L_f|e_j(t)|, ~~|\tilde{g}_j(e_j(t))|\le L_g|e_j(t)|
\end{align}

\begin{definition}
System (\ref{NNerror}) is said to achieve FnTSyn if there exists a finite time $T$, $e_{i}(t)=0$ for $t\geq T,i=1,2,\cdots,n$.
\end{definition}

At first, in order to use theoretical results in Section \ref{theory}, suppose $|e_i(t)|=\|e(t)\|_{\infty}$, we consider the derivative of $|e_i(t)|$,
\begin{align}\label{ajust}
&\frac{d|e_i(t)|}{dt}=\mathrm{sgn}(e_i(t))[-d_{i}e_{i}(t)\nonumber\\
&+\sum_{j=1}^na_{ij}\tilde{f}_j(e_j(t))+\sum_{j=1}^nb_{ij}\tilde{g}_j(e_j(t-\tau_{ij}(t)))]\nonumber\\
\le&(-d_i+L_f\sum_{j=1}^n|a_{ij}|)|e_i(t)|+L_g\sum_{j=1}^n|b_{ij}||e_j(t-\tau_{ij}(t))|\nonumber\\
\le&\max_i(-d_i+L_f\sum_{j=1}^n|a_{ij}|)\|e(t)\|_{\infty}\nonumber\\
&+L_g\max_i\sum_{j=1}^n|b_{ij}|\|e(t-\tau_{ij}(t))\|_{\infty}
\end{align}
therefore, according to Remark \ref{add}, we define
\begin{align}\label{alpha12}
\alpha_1=\max_i(-d_i+L_f\sum_{j=1}^n|a_{ij}|), \alpha_2=L_g\max_i(\sum_{j=1}^n|b_{ij}|)
\end{align}

The object is to design controllers $u_i(t)$ in (\ref{control}) under AIC and QC, and prove its validity to realize FnTSyn. We have
\begin{theorem}(Small delay case) \label{NNsmallthm}
For (\ref{NNerror}) under AIC (\ref{control}) and QC (\ref{zer}), suppose Assumptions \ref{as1}-\ref{sad} hold. The FnTSyn can be realized if conditions (\ref{condition1})-(\ref{condition6}) hold, where $\alpha_1$ and $\alpha_2$ is defined in (\ref{alpha12}).
\end{theorem}

Of course, adaptive technique can be also used, so we have
\begin{theorem}\label{NNadaptive}
For (\ref{NNerror}) under QC and AIC with adaptive strategy (\ref{controladaptive})-(\ref{alpha4adatpiverule}), if Assumptions \ref{as1}-\ref{sad} hold, and $\eta>(1-\overline{\gamma})\underline{\alpha}_3-\alpha_1$, where $\alpha_1$ and $\alpha_2$ is defined in (\ref{alpha12}), $\underline{\alpha}_3$ is defined in (\ref{xiajie1}), then FnTSyn can be realized.
\end{theorem}

\section{Numerical examples}\label{simu}
At first, let $\gamma=\frac{5}{19}$, $\gamma^{\prime}=\frac{1}{19}$, so $\overline{\gamma}=\frac{5}{19}$, according to (\ref{rho1rho2}),
\begin{align}
\rho_1=\frac{1-\gamma^{\prime}}{1+\gamma}=0.75,~~~\rho_2=\frac{1-\gamma}{1+\gamma^{\prime}}=0.7
\end{align}
by choosing $\pi_0=2$, we can define QC in (\ref{newquan}) as
\begin{align}\label{QCeg}
&\Pi=\{0\}\cup\{\cdots,\pm 7.2562, \pm 5.4422, \pm 3.8095, \pm 2.8571,\nonumber\\
&\pm 2, \pm 1.5, \pm 1.05, \pm 0.7875, \pm 0.5512, \pm 0.4134, \cdots\}
\end{align}
i.e., for any positive value $v$, if $v\in (\pi_{2l+1}, \pi_{2l})$, then
\begin{align}
q(v)=\left\{
\begin{array}{ll}
\pi_{2l+1}& v\le\frac{19}{18}\pi_{2l+1}\\
\pi_{2l}& \mathrm{otherwise}
\end{array}\right.
\end{align}
else, if $v\in (\pi_{2l}, \pi_{2l-1})$, then
\begin{align}
q(v)=\left\{
\begin{array}{ll}
\pi_{2l}& v\le\frac{19}{14}\pi_{2l}\\
\pi_{2l-1}& \mathrm{otherwise}
\end{array}\right.
\end{align}
else $q(v)=-q(-v), v<0$, and $q(0)=0$.

Next, AIC is defined as: for $k\ge 2$
\begin{align*}
t_{2k-1}=k-0.2+0.1\sum_{l=1}^{k-1}\sin(l),~ t_{2k}=k+0.1\sum_{l=1}^{k-1}\sin(l),
\end{align*}
where $t_{0}=0, t_1=0.8, t_{2}=1$, which is equivalent to
\begin{align}
&t_{2k+1}-t_{2k}=0.8+0.1\sin k\ge 0.7=\underline{\theta}, \\
&t_{2k+2}-t_{2k}=1+0.1\sin k\le 1.1=\overline{\theta}
\end{align}

Now, we choose the master NN as:
\begin{align}\label{mnn}
\dot{x}_i(t)=&-x_i(t)+\sum_{j=1}^3{a}_{ij}f(x_j(t))\nonumber\\
&+0.02\sum_{j=1}^3f(x_j(t-\tau_{ij}(t)))
\end{align}
where $x(t)=(x_1(t), x_2(t), x_3(t))^T$, $\tau_{ij}(t)=0.4+0.1\sin((i+2j)t)$, $f(x_i(t))=(|x_i(t)+1|-|x_i(t)-1|)/2$,
\begin{align*}
A=(a_{ij})=\left(
\begin{array}{ccc}
1.25&-3.2&-3.2\\
-3.2&1.1&-4.4\\
-3.2&4.4&1
\end{array}\right)
\end{align*}
So, $\tau_{ij}(t)\le \tau^{\bullet}=0.5<\overline{\theta}$ and $L_f=1$.

Let initial values $x(\kappa)=(0.05, -0.1, 0.15)^T, \kappa\in [-0.5, 0]$, its chaotic dynamical behavior can be found in Fig. \ref{chaos}.

\begin{figure}[h]
\begin{center}
\includegraphics[width=0.5\textwidth,height=0.27\textheight]{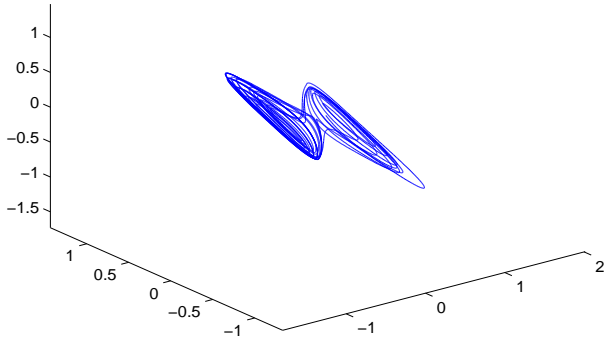}
\end{center}
\caption{Chaotic behavior of NN (\ref{mnn})}\label{chaos}
\end{figure}

Therefore, without external control, Syn cannot be realized, let alone FnTSyn. So, we add controller in the slave system:
\begin{align}\label{snn}
\dot{y}_i(t)=&-y_i(t)+\sum_{j=1}^3{a}_{ij}f(y_j(t))\nonumber\\
&+0.02\sum_{j=1}^3f(y_j(t-\tau_{ij}(t)))+u_i
\end{align}
where $y(\kappa)=(-1.5,0.8,-0.1)^T, \kappa\in [-0.5, 0]$, the controller $u_i$ under AIC and QC are defined in (\ref{control}).

According to (\ref{alpha12}), we have
\begin{align}
\alpha_1=\max_i(-1+\sum_{j=1}^3|a_{ij}|)=7.7, ~~\alpha_2=0.06
\end{align}

By simple calculations, one can get $\underline{\sigma}=15.52$ in (\ref{biyong1}), and $\underline{\varpi}_1=34.92$ in (\ref{xiajienew}), so $\alpha_3>\underline{\alpha}_3=222.4551$ in (\ref{xiajie1}) and $\alpha_4>\underline{\alpha}_4=34.98$ in (\ref{xiajie2}).

\begin{figure}
\begin{center}
\includegraphics[width=0.5\textwidth,height=0.27\textheight]{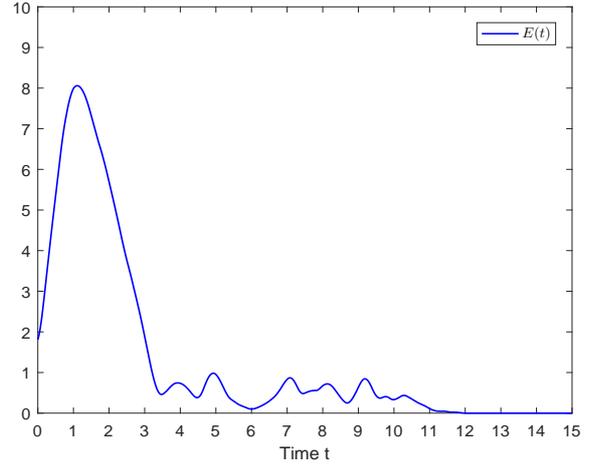}
\end{center}
\caption{Dynamics of $E(t)$ under constant control}\label{finitetime}
\end{figure}

Fig. \ref{finitetime} show that under the controller defined by (\ref{control}) with $\alpha_3=0.04$ and $\alpha_4=0.08$, FnTSyn can be realized, where synchronization error is defined as $E(t)=\|e(t)\|_2$. Of course, these values can be smaller than the above calculated theoretical lower bound.

Now, we apply the adaptive rules defined in (\ref{controladaptive})-(\ref{alpha4adatpiverule}), where $\alpha_3(0)=\alpha_4(0)=0$, and for $t_{2k}\le t\le t_{2k+1}$,
\begin{align}
\dot{\alpha}_3(t)=
\left\{
\begin{array}{ll}
0.01e^{0.2t}\|q(e(t))\|_{\infty};&\sup_{t-0.5\le s\le t}\|e(t)\|_{\infty}> 1\\
0.01\|q(e(t))\|_{\infty};&\sup_{t-0.5\le s\le t}\|e(t)\|_{\infty}\le 1
\end{array}
\right.
\end{align}
and
\begin{align}
\dot{\alpha}_4(t)=\left\{
\begin{array}{ll}
0;&\sup_{t-0.5\le s\le t}\|e(t)\|_{\infty}> 1\\
0.01;&0<\sup_{t-0.5\le s\le t}\|e(t)\|_{\infty}\le 1\\
0;&\sup_{t-0.5\le s\le t}\|e(t)\|_{\infty}=0
\end{array}
\right.
\end{align}

\begin{figure}
\begin{center}
\includegraphics[width=0.5\textwidth,height=0.27\textheight]{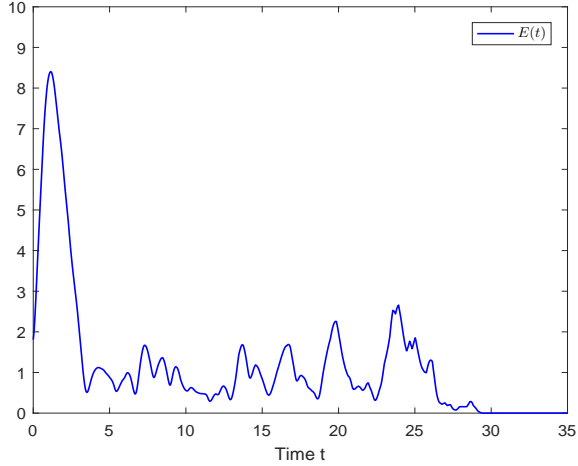}
\end{center}
\caption{Dynamics of $E(t)$ under adaptive control}\label{finitetimead}
\end{figure}

\begin{figure}
\begin{center}
\includegraphics[width=0.5\textwidth,height=0.27\textheight]{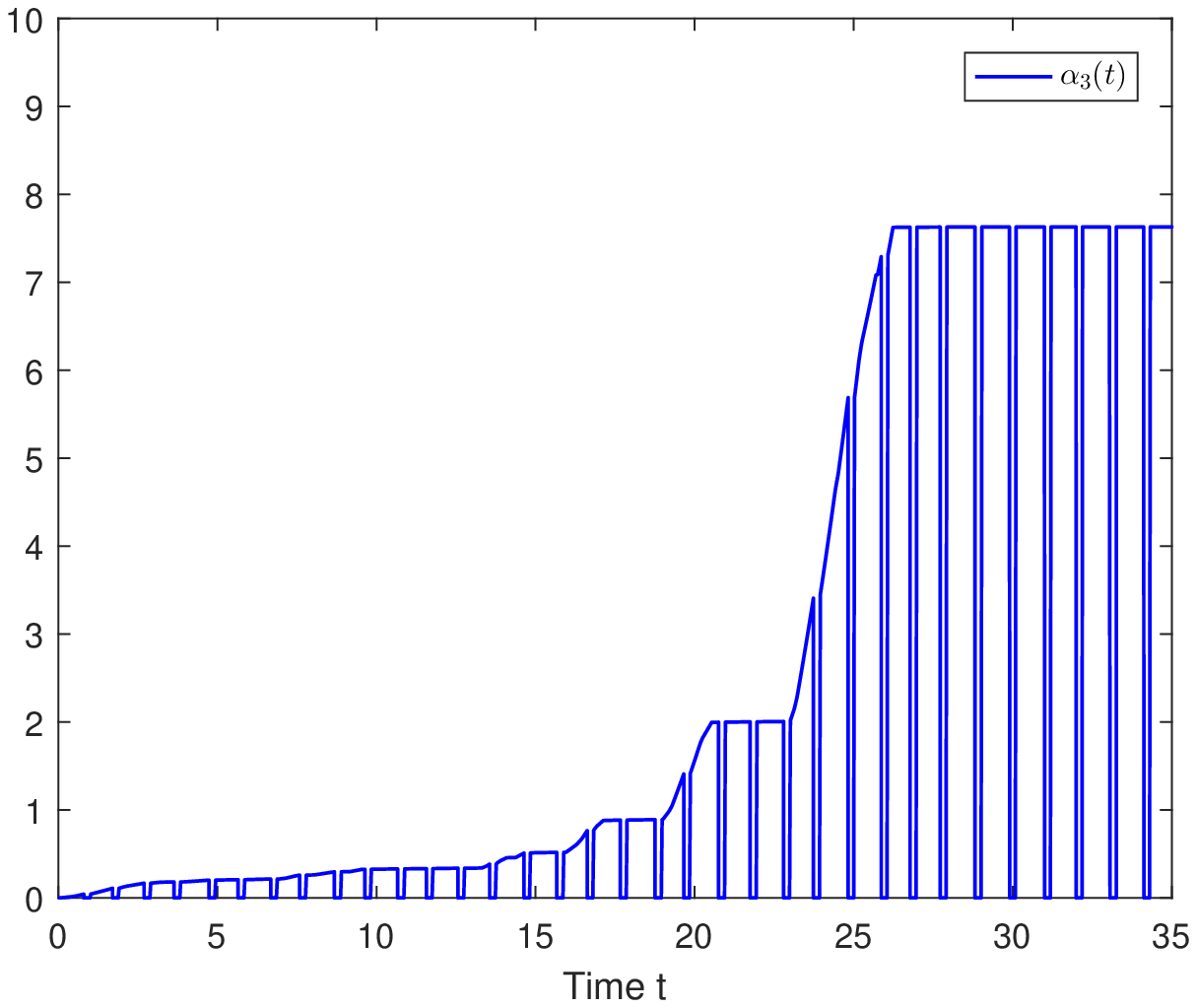}
\end{center}
\caption{Dynamics of $\alpha_3(t)$}\label{a3t}
\end{figure}

\begin{figure}
\begin{center}
\includegraphics[width=0.5\textwidth,height=0.27\textheight]{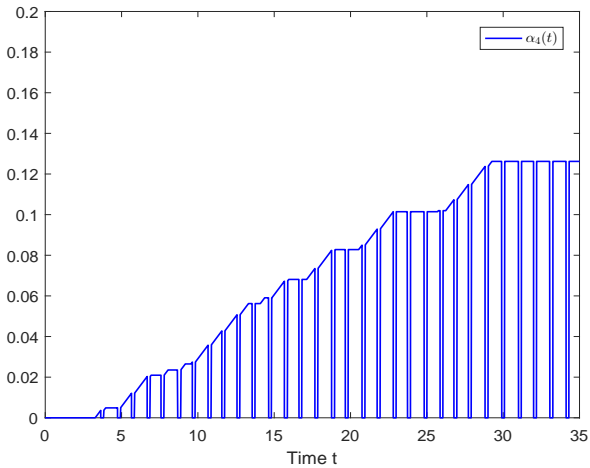}
\end{center}
\caption{Dynamics of $\alpha_4(t)$}\label{a4t}
\end{figure}

Fig. \ref{finitetimead} shows that FnTSyn can be finally achieved, and the dynamical behaviors of $\alpha_3(t)$ and $\alpha_4(t)$ can also be found in Fig. \ref{a3t} and Fig. \ref{a4t}.

\section{Conclusion}\label{con}
FnTSta of delayed systems via AIC and QC is investigated in this brief, where the bounded delay is required to be less than the infimum of CTS. AIC is a temporal control technique, and QC is a spatial control technique. A generalized QC is firstly defined. Then, using 2PM and $\infty$-norm, we set up the theories of FnTSta via AIC and QC. Moreover, the corresponding adaptive controller is also designed and proved. Finally, we apply obtained theories to FnTSyn of NNs.

Future researches include FnTSta with bounded large or unbounded delay and FxTSta with delay via AIC and QC.



%

\end{document}